\documentclass[preprint,5p,times,sort&compress,twocolumn]{elsarticle}
\usepackage{amsthm, amsfonts, amssymb}	%%
\usepackage{epstopdf}        			%%
\usepackage[fleqn]{amsmath}				%%
\usepackage{graphicx}					%%
\usepackage{txfonts}					%%
\usepackage{hyperref}					%%
\usepackage{lineno,color,comment}				
\newtheorem{theorem}{Theorem}[section]
\newtheorem{lemma}[theorem]{Lemma}

\newtheorem{corollary}[theorem]{Corollary}
\modulolinenumbers[5]					%%
\begin{document}

\begin{frontmatter}
\title{Majorization and dynamics of continuous distributions}
%\title{Hamiltonian poles and Lyapunov exponents in non Hermitian quantum dynamics}

\author[ufba]{Ignacio S. Gomez\corref{cor1}}
\ead{nachosky@fisica.unlp.edu.ar}
\author[if-sertao]{Bruno G. da Costa}
\ead{bruno.costa@ifsertao-pe.edu.br}
\author[ufrs]{M. A. F. dos Santos}
\ead{santos.maike@ufrgs.br}

%\author[iflp,unlp]{M. Portesi}
%\ead{portesi@fisica.unlp.edu.ar}
%\author[ufba]{Ernesto P. Borges}
%\ead{ernesto@ufba.br}

\cortext[cor1]{Corresponding author}
\address[ufba]{Instituto de F\'{i}sica, Universidade Federal da Bahia,
         Rua Barao de Jeremoabo, 40170-115 Salvador--BA, Brasil}
\address[if-sertao]{Instituto Federal de Educa\c{c}\~ao, Ci\^encia e Tecnologia do
             Sert\~ao Pernambucano,
             BR 407, km 08, 56314-520 Petrolina, Pernambuco, Brasil}
\address[ufrs]{Instituto de F\'isica, Universidade Federal do Rio Grande do Sul,
Caixa Postal 15051, CEP 91501-970, Porto Alegre, RS, Brazil}

%\address[iflp]{IFLP, CONICET, UNLP, Boulevard 113 e/63 y 64, 1900 La Plata, Argentina}
%\address[unlp]{Facultad de Ciencias Exactas, Universidad
%	Nacional de La Plata, C.C. 67, 1900 La Plata, Argentina}

\begin{abstract}
In this work we show how the concept
of majorization in continuous distributions
can be employed to characterize chaotic, diffusive and quantum
dynamics.
The key point lies in that majorization
allows to define an intuitive arrow of time,
within a continuous dynamics, along with
an associated \emph{majorized Second Law} which
implies the standard Second Law of thermodynamics
but not viceversa.
Moreover, mixing dynamics, generalized Fokker-Planck equations
and quantum evolutions are explored as
majorized ordered chains along the time evolution,
being the stationary states the infimum elements.
\end{abstract}


\end{frontmatter}

\nolinenumbers

\section{Introduction}
\label{intro}

The concept of majorization has shown a particular interest along
the last decades, mainly due to its
wide range of applicability in
information and quantum theory, among others.
The majorization is an operation
between a pair of finite-dimensional vectors that
gives a partial order in a
finite dimensional vectorial space.
When the finite vectors considered are discrete
probability distributions,
the majorization adopts the intuitive idea that, given two
discrete PDFs the distribution
which is majorized represents the probability vector more
spreading of the pair, and consequently it presents
the lowest Shannon entropy. Further developments
showed an intimately relation between majorization and
Schur-convex functions, from which subsequent
applications in quantum information protocols
turned out the majorization between two quantum states
an important criteria to be established [Nielsen].

Based in terms of the majorization of discrete finite-dimensional
vectors, Hardy, Littlewood and P\'{o}lya introduced the
continuous version for integrable functions, giving as result
a characterization of the convex ordering for random variables
in terms of majorization of vectors in the context of order statistics [Hardy].
Subsequent applications to stochastic orders in general were established from the
viewpoint of the continuous majorization [citacao].
In this sense, both types of majorization, the discrete and continuous ones,
bring different ways to relate the increasing direction of the majorization ordering
with the monotonic behavior (increasing or decreasing) of
convex functionals defined over probability distributions.

The goal of this letter is to formalize the relationship between the
increasing direction of the majorization ordering and the temporal evolution of
a continuous dynamics in order to characterize mixing dynamics, diffusion phenomena
and quantum evolutions. Thus, the present contribution also can shed light
towards a geometrical definition of the arrow time in the context of
convex functions.

The letter is structured as follows.
In Section~2 we review the concept of
continuous majorization of integrable functions, along
with some properties.
Then, in Section~3 we consider
a general motion equation
for a continuous probability distribution
and we establish necessary and sufficient conditions
for the set of time-parameterized distributions of a given
initial solution is an ordered chain by majorization.
From this
%characterization
%we give a general expression
%for the evolution operator, being
we characterize the stationary and the initial states
as the infimum and the supreme ones
of all ordered chain by continuous
majorization. Moreover, a generalized
Second Law of thermodynamics in the context of majorization,
is proposed.
Section~4 is devoted to illustrate
the scope of the formalism presented.
First, we consider a continuous dynamical system and we obtain a
necessary condition for mixing in terms of majorization.
Next, we characterize general Fokker-Planck equations
(without drift terms) as totally
ordered chains (by time) of probability distributions, and
their associated Fisher information are also obtained.
Last, quantum evolutions (unitary and non-unitary) are explored
from the viewpoint of continuous majorization.
Finally, in Section~5 some conclusions and perspectives are outlined.

%\section{Preliminaries}\label{sec prelim}

\section{Majorization of integrable functions and discrete vectors}\label{sec prelim}
Here we give the necessary elements for the development of the present work.
We begin by recalling the concept of continuous majorization along with some properties.
\subsection{Continuous majorization}
We say that a real function $f(x)$ is convex (concave respectively) on a real interval $I$
if for all $x,y\in I$ we have $f(\alpha x+\beta y)\leq \alpha f(x)+\beta f(y)$ ($\geq$ respectively).
Consider the set $L^1((0,1))$ of all real Lebesgue integrable function on $(0,1)$
and denote by $\mathcal{L}_{\textrm{cx}}(I)$ the set of all convex functions on $I$.
Given $f,g:I\rightarrow \mathbb{R} \in L^1((0,1))$ it is said that $f$ is \emph{majorized}
by $g$, denoted by $f\prec g$, iff [Hardy]

\begin{equation}\label{cont-majorization-def}
\int_{0}^{1}\phi(f(t))dt \leq \int_{0}^{1}\phi(g(t))dt
\quad , \quad \forall \ \phi\in \mathcal{L}_{\textrm{cx}}(I)
\end{equation}
whenever the integrals exist.
When $\mathcal{L}_{\textrm{cx}}(I)$ is substituted
by $\mathcal{L}_{\textrm{icx}}(I)$ (denoting the increasing convex functions on $I$)
then the partial order obtained is \emph{weak majorization}, symbolized by $\prec_W$.
It can be seen that $\prec$ is a reflexive and transitive relation
\footnote{I.e. for all $f,g,h\in L^1((0,1))$ we have
$f\prec f$ and $f\prec g,g\prec h\Rightarrow f\prec h$.} in $L^1((0,1))$. Moreover,
$\prec$ is a partial order since $f\prec g$ and $g\prec f$ do not imply
necessarily $f=g$ a.e.
If a set of distributions $\mathcal{D}$ has two elements $g,h$ such that
$g\prec f \prec h$ for all $f\in\mathcal{D}$ then it is said that
$g$ and $h$ are the \emph{infimum} and the \emph{supreme} of $\mathcal{D}$.
We say that $\mathcal{D}$ is an \emph{ordered chain} by majorization
if for all $f,g\in\mathcal{D}$ we have $f\prec g$ or $g\prec f$.

\subsection{Discrete majorization}

Complementary, Hardy also defined the discrete majorization and showed
its relation with the continuous case, as follows.
Let $\mathbf{x},\mathbf{y}\in \mathbb{R}^n$ two
$n$-dimensional vectors. Then we say that $\mathbf{x}\prec \mathbf{y}$ iff
\begin{eqnarray}\label{disc-majorization-def}
&S_{k:n}(\mathbf{x})\leq S_{k:n}(\mathbf{y}), \quad 1\leq k\leq n, \quad
\textrm{with} \ S_{n:n}(\mathbf{x})=S_{n:n}(\mathbf{y})\nonumber\\
& \textrm{where} \quad S_{k:n}(\mathbf{z})=\sum_{i=n+1-k}^{n}z_{i:n} , \quad \mathbf{z}\in\mathbb{R}^n
\end{eqnarray}
Here $z_{1:n}\leq z_{2:n}\leq \ldots \leq z_{n:n}$ denote the components of $\mathbf{z}$
arranged in increasing order. As in the continuous case, the discrete majorization
is a partial order in the set $\{\mathbf{z}\in \mathbb{R}^n: z_i\leq z_{i+1} i=1,\ldots,n-1\}$
since from $\mathbf{x}\prec \mathbf{y}$ and $\mathbf{y}\prec \mathbf{x}$
it follows only
that $\mathbf{y}$ is a permutation of $\mathbf{x}$.

The relationship between the continuous majorization
and the discrete one is
described by the following result [].
\begin{theorem}
Let $\mathbf{x},\mathbf{y}\in \mathbb{R}^n$ two
$n$-dimensional vectors and $I\subseteq\mathbb{R}$ an interval.
Then the following propositions
are equivalent:
\begin{itemize}
  \item[$(a)$]  $\mathbf{x}$ is majorized by $\mathbf{y}$.
  \item[$(b)$] $\phi(\mathbf{x})\leq \phi(\mathbf{y})$ for each Schur-convex
  function $\phi:I^n\rightarrow \mathbb{R}$.
  \item[$(c)$] $\phi(\mathbf{x})\leq \phi(\mathbf{y})$ for each symmetric
  quasi-convex function ~$\phi:I^n\rightarrow \mathbb{R}$.
  \item[$(d)$] $\sum_{i=1}^{n}\phi(x_i)\leq \sum_{i=1}^{n}\phi(y_i)$
  for each convex function ~$g:I\rightarrow\mathbb{R}$.
\end{itemize}
\end{theorem}
Note that $(d)$ is the discrete version of the continuous
majorization \eqref{cont-majorization-def}. If necessary, this is
the equivalence that we will use.

\subsection{Probability distributions and majorization}

When the functions or vectors represent probability distributions
the majorization adopts an intuitive interpretation.
For illustrate this, we consider the discrete version and the same
argument can be applied to the continuous case.

Let
$\mathbf{p}=(p_1,\ldots,p_n)$ and $\mathbf{q}=(q_1,\ldots,q_n)$
be two $n$-dimensional probability distributions,
i.e. $p_i,q_i\geq 0$ for all $i=1,\ldots,n$ and
$\sum_{i=1}^{n}p_i=\sum_{i=1}^{n}q_i=1$. It is clear that
$\frac{1}{n}\leq p_i,q_i\leq 1$ for all $i=1,\ldots,n$ so by the
definition \eqref{disc-majorization-def} we have that
\begin{equation}\label{infimum-supreme}
(\frac{1}{n},\ldots,\frac{1}{n})
\prec \mathbf{p},\mathbf{q} \prec (1,\ldots,0)
\end{equation}
where $(\frac{1}{n},\ldots,\frac{1}{n})$ and $(1,\ldots,0)$
are the uniform and the delta distributions respectively.
This means that, besides $\prec$ is a partial order on the
space of $n$-dimensional probability vectors, the uniform and the
delta distributions are the infimum and the supreme elements
respectively.

Physically, if we have an experiment with $n$ possible
outcomes $x_i$ represented by the space of events $\Gamma=\{x_1,\ldots,x_n\}$
and $\mathbf{p}$ is a probability vector
(with $p_i$ the probability of that $i$th outcome occurs) then discrete
majorization says that the
distribution with the minimal information (the uniform one)
is majorized by $\mathbf{p}$, and in turn,
$\mathbf{p}$ is majorized by the maximal information distribution (the delta one).
In the next sections we exploit this idea for characterizing temporal evolutions of
continuous distributions from majorization.

\section{Temporal evolution of continuous distributions from
majorization}

In this section we consider a system described by a continuous
distribution $p(x,t)$ containing its maximal information about the dynamics
at time $t$ where $x$ represents a continuous variable in $(0,1)$.
We focus our study in the relationship between the dynamics and the
the majorization \eqref{cont-majorization-def} restricted to
the set of time-parameterized distributions $\mathcal{P}=\{p_t: t\geq 0\}$ with
$p_t$ the distribution of the system at time $t$, i.e. $p_t=p(x,t)$.
Given an arbitrary dynamics, the set $\mathcal{P}$ represents the evolution
of the system from $t=0$ to $t=\infty$.
%Since the functions employed in physics are typically differentiable,
%without loss of generality we assume that the convex functions $\phi$
%in the continuous majorization are also differentiable\footnote{This condition
%is not restrictive since differentiable functions in
%$\mathcal{L}_{\textrm{cx}}(I)$ are
%dense in $\mathcal{L}_{\textrm{cx}}(I)$.}
A first consequence from continuous majorization applied to
$\mathcal{P}=\{p_t: t\geq 0\}$ is given by the following Lemma.
\begin{lemma}\label{lemma-ordered-chain}
The following statements are equivalent:
\begin{itemize}
    \item [$(I)$] $\mathcal{P}=\{p_t: t\geq 0\}$ is an ordered chain by majorization
with $p_{t_2}\prec p_{t_1}$ for all $t_1\leq t_2$ (i.e. the
distribution at time $t$ is majorized by all the precedent ones).
  \item [$(II)$] The function $\lambda_{\phi}(t):[0,\infty]\rightarrow \mathbb{R}$,
  $\lambda_{\phi}(t)=\int_{0}^{1}\phi(p_t(x))dx$ is decreasing for all
  $\phi\in \mathcal{L}_{\textrm{cx}}(I)$.
\end{itemize}
In turn, $(I)$ or $(II)$ imply
that $\lambda_{\phi}'(t)=\int_{0}^{1}\phi'(p_t(x))(\partial p_t/\partial t)dx\leq 0$
for all differentiable $\phi\in \mathcal{L}_{\textrm{cx}}(I)$, and that
the initial distribution $p_0$ is the supreme of $\mathcal{P}$.
\end{lemma}
\begin{proof}
It follows directly from the definition
\eqref{cont-majorization-def} applied to $\mathcal{P}$.
\end{proof}
The content of Lemma \ref{lemma-ordered-chain}
is that when we have a dynamics satisfying $(I)$ then
this can be characterized by the behavior increasing of the
functions $\lambda_{\phi}(t)=\int_{0}^{1}\phi(p_t(x))dx$ for all convex function
$\phi$, where
the initial distribution majorizes all the subsequent evolved ones.
Thus, a first simple connection between dynamics and continuous
majorization is provided. We shall see that the hypothesis $(I)$
is compatible with the intuitive idea of that, in diffusion phenomena,
as the distribution evolves it tends to spread along
its domain (thus approaching to the uniform one that is
the infimum element). This is precisely the content
of the next result.
\begin{lemma}\label{lemma-asymptotic}
Assume the hypothesis $(I)$ of Lemma \ref{lemma-ordered-chain}
and the existence of an asymptotic distribution
at $t\rightarrow\infty$, denoted by $p_{\infty}$,
satisfying $\partial p_{\infty}/\partial t=0$ then we have
\begin{equation}\label{lemma2-1}
p_{\infty} \prec p_t \prec p_0 \quad , \quad \forall \ t\geq 0
\end{equation}
and
\begin{equation}\label{lemma2-2}
\lambda_{\phi}'(p_{\infty})=0 \quad , \quad \forall \ \phi\in\mathcal{L}_{\textrm{cx}} \ \textrm{differentiable}
\end{equation}
\end{lemma}
The relevance of Lemmas \ref{lemma-ordered-chain}
and \ref{lemma-asymptotic} is that, when
the time evolution preserves majorization in the sense
of condition $(I)$, they allow to characterize
the dynamics in terms of
continuous majorization only, where the initial state
is the supreme and the stationary one is the infimum.
In particular,
%we can analyze the behavior of the
for the
Shannon-Gibbs entropy functional $S[p]=-\int \psi(p(x))dx$ given by the convex function
$\psi(x)=x\ln x$ the condition $(II)$ says that
$-\lambda_{\psi}(t)$ is an increasing function of $t$ so
$-\lambda_{\psi}(t_1)\leq -\lambda_{\psi}(t_2)$ $\forall t_1\leq t_2$,
i.e. $S[p_{t_1}] \leq S[p_{t_2}]$ $\forall t_1\leq t_2$, in accordance
with the Second Law of thermodynamics.
%Encouraged by the description of the Second Law as a particular
%case for $\psi(x)=x\ln x$ under the condition of a dynamics
%satisfying $(I)$,
The following majorized version of the
Second Law is obtained from Lemma \ref{lemma-ordered-chain}:
\begin{corollary}\label{majorizes-second-law}
If the dynamics is preserved by majorization
in the sense of condition $(I)$ of Lemma \ref{lemma-ordered-chain}
then a \textbf{Majorized Second Law (MSL)} is satisfied
\begin{equation}\label{majorized-SL}
\Delta \lambda_{\phi}=\lambda_{\phi}(t_2)-\lambda_{\phi}(t_1)\leq 0 \quad \, \quad \forall t_1\leq t_2
\end{equation}
for all $\phi\in \mathcal{L}_{\textrm{cx}}(I)$. In particular,
for $\phi(x)=x\ln x$ the standard Second Law is recovered. Thus, we also have
\begin{equation}\label{MSL-SL}
\textrm{Majorized Second Law (MSL)} \quad \Longrightarrow \quad \textrm{Second Law (SL)}
\end{equation}
\end{corollary}
Interestingly, the equation \eqref{majorized-SL} defines an arrow of time,
composed by events that fulfill the condition $(I)$
along the time, which must be satisfied by all convex function
$\phi\in \mathcal{L}_{\textrm{cx}}(I)$. This will be discussed with some detail
in the forthcoming section.
Moreover, the implication
\eqref{MSL-SL} expresses that MSL is stronger than SL.
Next step is to study what kind of phenomena can be compatible with MSL.

\section{Applications}

In order to study
%test the reliability of the MSL and
what type
of dynamics can be obtained from continuous majorization, in this section
we illustrate the results with some examples. Then, we
come back to discuss an interpretation of the arrow of time
given by majorization.

\subsection{Dynamical systems: mixing}

One of the central concepts of dynamical systems theory
and statistical mechanics is the mixing condition, i.e. the
asymptotic vanishing of the correlations
between two subsets of phase space that are sufficiently separated in time.
In its usual definition in the language of distributions this reads as
\begin{eqnarray}\label{mixing}
&\exists f_{*}\in L^1(X) \ \textrm{such that} \ \forall f\in L^1(X),g\in L^{\infty}(X):\nonumber\\
&\lim_{t\rightarrow\infty}\int_X f(T_t(x))g(x)dx=\int_X f_*(x)g(x)dx
\end{eqnarray}
where $X$ is the phase space, $f_{*}$
is the equilibrium distribution of the system at $t\rightarrow\infty$
($f_{*}\circ T_t=f_{*}$), and
$T_t:X\rightarrow X$ is a continuous transformation, typically the
Liouville time evolution in classical mechanics. In particular, Eq. \eqref{mixing}
says that the measure $\mu_{*}(A)=\int_A f_{*}(x)dx$ is invariant under $T_t$.

Now assume the dynamics satisfies the condition $(I)$
of Lemma \ref{lemma-ordered-chain} and $X=(0,1)$.
In particular, for the convex functional
$\phi(x)=|x|$ we have that $||f\circ T_{t} ||_1\geq ||f\circ T_{t^{\prime}} ||_1$
for all $t\leq t^{\prime}$ so by Lebesgue dominated converge it follows that
the $\lim_{n\rightarrow\infty}||f\circ T_{t_n} ||_1=||f_{*}||_1$ exists
for an increasing sequence $t_1<t_2<\ldots<t_n<\ldots$ and $f_*\in L^1$.
In turn, this implies \eqref{mixing} and thus the system is mixing.

\subsection{Generalized Fokker-Planck equations}

Disordered and thermal molecular motion is macroscopically
characterized as diffusion phenomena of a net flux of particles from one
region to other. Under Markovian assumptions and making the passing to the continuum,
the discretized master
equation for the probability transition states becomes into the
Fokker-Planck equation (FPE).
Recently, a generalization of the FPE
(recovering the nonlinear and linear cases as special ones),
that links generalized entropic forms with the theorem $H$, was proposed in the form
[Curado, Nobre]
\begin{equation}\label{curado-nobre}
\frac{\partial p(x,t)}{\partial t}=-\frac{\partial \{F(x)\Psi[p(x,t)]\}}{\partial x}+
\frac{\partial}{\partial x}\left\{\Omega[p(x,t)]\frac{\partial p(x,t)}{\partial x} \right\}
\end{equation}
where $p(x,t)$ is the probability distribution
of the particles at time $t$, $F(x)=-\frac{d\varphi}{dx}$ is a conservative force
acting over the particles,
and $\Omega[p],\Psi[p]>0$ are nonnegative functionals.
From the functional $\lambda_{\phi}(t)$ we can relate majorization with the generalized
FPE \eqref{curado-nobre} as follows.
%Assuming the condition $(I)$ we have that $\lambda'_{\phi}(t)$ for all
%$\phi\in\mathcal{L}_{\textrm{cx}}(I)$ so
Integrating by parts we obtain
\begin{eqnarray}\label{majorization-FPE}
&\lambda'_{\phi}(t)=\left[\phi'(p_t(x))\int^{x}
ds\frac{\partial p(s,t)}{\partial t}\right]_{0}^1 \nonumber\\
&-\int_{0}^{1}\phi''(p_t(x))\frac{\partial p(x,t)}{\partial x}
\left(\int^{x}ds\frac{\partial p(s,t)}{\partial t}\right)dx
\end{eqnarray}
where the first term can be neglected (vanishing of the functions
in the extremes $x=0,1$) and a reasonable assumption is
to make $\int^{x}ds\frac{\partial p(s,t)}{\partial t}=\Omega[p(x,t)]\frac{\partial p(x,t)}{\partial x}$
with $\Omega[p(x,t)]>0$.
Replacing this in \eqref{majorization-FPE} and using that $\phi''\geq0$
we conclude
\begin{eqnarray}\label{majorization-FPE2}
\lambda'_{\phi}(t)=
-\int_{0}^{1}\phi''(p_t(x))\Omega[p(x,t)]\left(\frac{\partial p(x,t)}{\partial x}\right)^2
dx\leq 0
\end{eqnarray}
for all $\phi\in\mathcal{L}_{\textrm{cx}}(I)$ differentiable.
Hence, the generalized FPE \eqref{curado-nobre} with $F(x)=0$ satisfies
the condition $(I)$ of Lemma \ref{lemma-ordered-chain} and thus
the MSL is recovered.

\subsection{Quantum dynamics}

We analyze how continuous majorization can characterize
quantum dynamics.
For the sake of simplicity we consider that the set $\mathcal{P}$ is given by
the evolution of the eigenfunctions probability distributions, i.e.
$\mathcal{P}_{n}=\{|\psi_n(x,t)|^2: t\geq0\}$ with $n$ the energy index,
\begin{equation}\label{hamiltonian-dynamics}
i\hbar \frac{\partial \psi_n}{\partial t}=H \psi_n=E_n\psi_n
\end{equation}
and,
\begin{equation}\label{hamiltonian-energies}
E_n=\epsilon_n + i\gamma_n
\end{equation}
Eq. \eqref{hamiltonian-energies} expresses that
the Hamiltonian $H$ is not necessarily Hermitian,
typical of an open system dynamics. In this non-Hermitian
context, the measurable eigenergies of the system
are the $\epsilon_n$ while the $|\gamma_n|^2$ are proportional to the decay times.
It is clear that the usual unitary case is recovered when $\gamma_n=0$
for all $n$.
In order to verify if the dynamics \eqref{hamiltonian-dynamics}
preserves majorization in the sense of the condition $(I)$
of Lemma \ref{lemma-ordered-chain}, we calculate the derivate of
the function $\lambda_{\psi_n}(t)$ for each energy index $n$.
Doing this and using that
$\frac{d\psi_n}{dt}\psi_n^{*}+\frac{d\psi_n^{*}}{dt}\psi_n=2\frac{\gamma_n}{\hbar}|\psi_n|^2$
we obtain
\begin{equation}\label{quantum-majorization}
\lambda_{\phi}'(t)=\frac{2\gamma_n}{\hbar}\int_{0}^{1}\phi'(|\psi_n|^2)|\psi_n|^2dx
\  , \  \forall \ \phi\in\mathcal{L}_{\textrm{cx}}(I) \ \textrm{differentiable}
\end{equation}
where the domain of the variable $x$ of the eigenfunctions
$\psi_n(x,t)$ is assumed to be $(0,1)$.
Equation \eqref{quantum-majorization} is the starting point for
characterizing some types of quantum dynamics.

\emph{Case I: Hermitian dynamics $\gamma_n=0 \ \forall n$:}
From \eqref{quantum-majorization} we can see that
$\lambda_{\phi}'(t)$ for all $\psi_n$, which implies
that for all $\psi_n(x,t_1),\psi_n(x,t_2)$ and $t_1,t_2$
we have $|\psi_n(x,t_1)|^2 \prec |\psi_n(x,t_2)|^2$
and $|\psi_n(x,t_2)|^2 \prec |\psi_n(x,t_1)|^2$.
This means that the infimum and the supreme are always the
same $|\psi_n(x,t)|^2$ (with $t\geq0$ arbitrary)
along time, as expected for the Hamiltonian eigenstates
in an unitary dynamics. Accordingly, in this case MSL
gives that $\Delta \lambda_{\phi}=0$, as the system evolves.

\emph{Case II: non-Hermitian dynamics $\gamma_j\neq0$:}
Given $\phi\in\mathcal{L}_{\textrm{cx}}(I)$ differentiable,
since $\phi$ is convex then $\phi''$ is nonnegative
which implies that $\phi'$ is increasing. In particular,
from $|\psi_n(x,t)|^2\geq 0$ it follows that
$\phi'(|\psi_n|^2)\geq \phi'(0)$. Assuming
$\gamma_j<0$ for some $j$ we obtain
\begin{equation}\label{quantum-majorization-nonhermitian}
\lambda_{\phi}'(t)
\leq \frac{2\gamma_j}{\hbar}\phi'(0)
\quad  , \quad   \forall \ \phi\in\mathcal{L}_{\textrm{cx}}(I) \ \textrm{differentiable}
\end{equation}
In this case we see that, in order to
fulfill the condition $(I)$ we need $\phi'(0)\geq0$
($\phi'(0)\leq0$ if $\gamma_j>0$ respectively)
for all $\phi\in\mathcal{L}_{\textrm{cx}}(I)$, which
can be satisfied if $\phi\in\mathcal{L}_{\textrm{icx}}(I)$, thus
leading to a weak majorization.
Hence, for the case of a non-Hermitian dynamics we have that
$\mathcal{P}_{j}=\{|\psi_j(x,t)|^2: t\geq0\}$ is a ordered
chain by weak majorization when $\gamma_j<0$.

Moreover, from \eqref{quantum-majorization} an interplay between the sign of $\gamma_j$
and the Second Law can be depicted easily as follows.
Specializing \eqref{quantum-majorization} with
$\phi(x)=x\ln x$

\subsection{Arrow of time and majorization}

The problem of the arrow of time has found
partial solutions
%from the cosmological viewpoint
that are
linked with thermodynamics by means of the irreversibility
of the processes outside of equilibrium.
Thus, from a cosmological viewpoint the arrow of time
must be consistent with the standard models of the universe evolution,
for instance the Big Bang model precludes the expansion from a spatial-temporal
point of zero volume towards the current
universe structured by galaxies, stars, quasars
and so on. The underlying element behind the arrow time is the
irreversibility which is characterized
in terms of a functional of the state of the system,
the thermodynamical entropy.
By the second law of thermodynamics, the
variation of the thermodynamical entropy of a system
of must be non-negative. Interestingly, if we consider
the convex character
of the entropy functional, we can see that
continuous majorization offers a generalized
second law for distributions governed by evolutions
given by the Theorem. Even more, majorization
indicates that all convex functional must obey
second law
when an ordered chain of evolved distributions is
considered, as
we can see from Lemma \ref{lemma-ordered-chain}. This statement
results more stronger than the standard
second law referred only to the entropy functional,
due to the fact that the set of all states
(expressed by probability distributions) that one
can consider physically is partially ordered by majorization.
In other words, the second law concerning
only the thermodynamical entropy allows the possibility
of violation of majorization and thus a given initial
state not necessarily can be majorized by an evolved one,
and viceversa.
From all this we conclude that the arrow of time given by majorization
express only a portion of all the possible sequences of distributions
evolved along time, the ones that can be totally majorized by the time,
and therefore is expected to find physical dynamics not characterized by majorization.

\section{Conclusions}\label{sec conclusions}

Assuming the existence of a finite characteristic time scale and a mixing phase space, we have presented a method for obtaining the Kolmogorov--Sinai entropy of a quantum system as a function of the time scale.
The three ingredients that we used were: 1) the natural graininess of the quantum phase space given by
the Uncertainty Principle,
2) a time rescaled KS--entropy that allows one to introduce the
characteristic time scale of the system as a parameter, and 3) a mixing condition at the (finite) characteristic time scale.

In summary, our contribution is two--fold. On the one hand, the correspondence between classical and
quantum elements of the mixing formalism provides a framework for exporting theorems and results of
the classical ergodic theory to quantum language %(lemmas \ref{lemma2} and \ref{lemma3})
which is
schematized in Fig. \ref{fig:classical-quantum}. On the other hand, the previous steps for obtaining
the equation (\ref{time scale1})
can be considered as a rigorous proof of the existence of the logarithmic time scale when the dynamics
 in quantum phase space is mixing at a finite time, thus providing a theoretical bridge between the ergodic
theory and graininess of quantum mechanics.

Finally, it is worth noting that a non--zero Kolmogorov-Sinai entropy does not necessarily
require strong classical mixing properties. Since in the present work it is assumed
a mixing at finite time dynamics, this precisely implies a mixing one.
When this assumption is not satisfied, one could suggest that the dependence of the Ehrenfest--time does
not need to be logarithmical in $q$ since regular and weak chaotic systems obey
a power law in their correlation decays.

As was accomplished in \cite{paper0,paper1,paper2,paper3,paper4,paper5}, we hope that the use of more
results of the ergodic hierarchy allows to shed light on quantum chaos
theory in future researches.

\section*{Acknowledgments}
We specially acknowledge the comments of the reviewers which helped us to improve the present contribution.
This work was partially supported by CONICET and Universidad Nacional de La Plata, Argentina.
% -\0p

\end{document}